\documentclass[11pt,a4 paper]{article}
\newtheorem{theorem}{Theorem}[section]

\newtheorem{definition}{Definition}[section]
\newtheorem{lemma}{Lemma}[section]

\usepackage{authblk}
\usepackage{multirow,bm}
\usepackage{hhline}
\usepackage{array}
\usepackage{makecell}
\usepackage{hyperref}
\hypersetup{hypertex=true, colorlinks=true, linkcolor=red, anchorcolor=red, citecolor=red}

\usepackage{float}
\usepackage{tocloft}
\usepackage{amsfonts,amsmath,amssymb}
\usepackage{extarrows}
\usepackage{multirow, verbatim}
\usepackage{shadow,enumerate}
\usepackage{makeidx}  
\usepackage{graphicx}
\usepackage{color}

\usepackage{mathrsfs}

\def\whitebox{{\hbox{\hskip 1pt
        \vrule height 6pt depth 1.5pt
        \lower 1.5pt\vbox to 7.5pt{\hrule width
                  3.2pt\vfill\hrule width 3.2pt}%
        \vrule height 6pt depth 1.5pt
        \hskip 1pt } }}
\def\qed{\ifhmode\allowbreak\else\nobreak\fi\hfill\quad\nobreak\whitebox\medbreak}

\begin{document}
\title{The 4-adic complexity of quaternary sequences with low autocorrelation and high linear complexity}
\author[a]{Feifei Yan}
\author[b]{Pinhui Ke\thanks{Corresponding author: keph@fjnu.edu.cn. This work is supported by the National Natural Science Foundation of China (No.62272420), Natural Science Foundation of Fujian Province (No. 2023J01535).}}
\author[c]{Lingmei Xiao}
\affil[a]{\it{\small{School of Mathematics and Statistics, Fujian Normal University, Fuzhou, Fujian, 350117, P. R. China}}}
\affil[b]{\it{\small{Key Laboratory of Analytical Mathematics and Applications (Ministry of Education), Fujian Normal University, Fuzhou, Fujian, 350117, P. R. China}}}
\affil[c]{\it{\small{Provincial Key Laboratory of Applied Mathematics, Putian University, Putian, Fujian, 351100, P. R. China}}}
\date{\today}
\maketitle

\begin{abstract}
Recently, Jiang et al. proposed several new classes of quaternary sequences with low autocorrelation and high linear complexity by using the inverse Gray mapping (JAMC, \textbf{69} (2023): 689--706). In this paper, we estimate the 4-adic complexity of these quaternary sequences. Our results show that these sequences have large 4-adic complexity to resist the attack of the rational approximation algorithm.

\end{abstract}

\section{Introduction}
Pseudo-random sequences are widely used in many fields, such as modern communication, software testing, radar system, coding theory, stream cipher systems, and so on \cite{w1,w2}. In practical applications, sequences are required to have low correlation and high complexity. For the sake of simplify implementation, binary and quaternary sequences with good properties are preferred sequences in many applications.

Theoretically, quaternary sequences can be generated by linear feedback shift registers(LFSRs) \cite{w3} or feedback with carry shift registers(FCSRs) \cite{w4,w5}. The shortest lengths of LFSR and FCSR that can generate a given quaternary sequence are called linear complexity and 4-adic complexity, respectively. Compared to the linear complexity, research on the 4-adic complexity of quaternary sequences has not been fully developed. The concept of 4-adic complexity was proposed by Klapper, Goresky and Xu in the 1990's \cite{w5,w6}. In order to resist the attack of the rational approximation algorithm(RAA), the 4-adic complexity of a quaternary sequence with period $N$ should exceed $\frac{N-16}{6}$ \cite{w5}. Recently,  some works have been devoted to studying the 4-adic complexity of sequences \cite{w7,w8,w9,w10,w11,w12,w13}. Yang et al. proved the 4-adic complexity of some even period quaternary sequences with ideal autocorrelation, and the results show that the 4-adic complexity of these quaternary sequences are large enough to resist the attack of the RAA \cite{w7,w8}. Edemskiy et al. considered the symmetric 4-adic complexity of several odd period quaternary sequences with low autocorrelation, and showed that these sequences have a sufficiently large 4-adic complexity to resist the attack of the RAA \cite{w10,w11}.

Quaternary sequences can be constructed by using the inverse Gray mapping \cite{w14,w15,w16}. Recently, Jiang et al. constructed some new quaternary sequences by using the inverse Gray mapping, which have low correlation and high linear complexity \cite{w14}. In this paper, we will estimate the 4-adic complexity of these sequences.

The rest of this paper is organized as follows. In Section 2, we recall some related concepts and the required lemmas. In Section 3, we estimate the 4-adic complexity of quaternary sequences proposed in \cite{w14}. And the last section is the conclusion of this paper.
\section{Preliminaries}
In this section, we will introduce some notations and concepts which will be used later.

\subsection{Cyclotomic classes and cyclotomic numbers}
Let $p$ be an odd prime number with $p=4f+1=x^{2}+4y^{2}$, where $f$ is an odd integer, $y$ is an integer, and $x\equiv 1\pmod 4$. Let $\theta$ be a primitive element of $\mathbb{F}_{p}$. Define $D_{i}=\left\lbrace \theta^{i+4j}:j=0,1,2,\dots, f-1\right\rbrace $, where $i=0,1,2,3$. $D_{i}$ is called the cyclotomic classes of order 4 in $\mathbb{F}_{p}$. It is easy to verify that $\mathbb{F}_{p}=\left\lbrace 0\right\rbrace \cup D_{1}\cup D_{2}\cup D_{3}\cup D_{4}$.
\begin{definition}\label{def1}\cite{w17}
	The cyclotomic numbers of order 4 in $\mathbb{F}_{p}$ are defined as
	\begin{equation}
		(i,j)_{4}=|(D_{i}+1)\cap D_{j}|,0\le i,j\le3.\notag
	\end{equation}
\end{definition}
The formula for calculating cyclotomic numbers of order 4 was given as follows.
\begin{lemma}\label{lem1}\cite{w18}
	Let $p=x^{2}+4y^{2}$, $x,y\in\mathbb{Z}$ and $x\equiv1\pmod4$. When $p\equiv5\pmod8$, the relationship and formula of these cyclotomic numbers are shown in Table 1,
\begin{center}
		\begin{table}
		\centering{\caption{}{$\mathrm{\;The \;relationship\;between\;cyclotomic \;numbers}$ when $p\equiv5\pmod8$}}			
\begin{center}
	\begin{tabular}{l c c c r}
		\hline
		$(i,j)_{4}$ & $0$ & $1$ & $2$ & $3$\\
		\hline
		$0$ & A & B & C & D \\
		$1$ & E & E & D & B \\
		$2$ & A & E & A & E \\
		$3$ & E & D & B & E \\
		\hline
	\end{tabular}
\end{center}
\end{table}
\end{center}
where
\begin{equation}
	A=\frac{p-7+2x}{16}, B=\frac{p+1+2x-8y}{16}, C=\frac{p+1-6x}{16},\notag
\end{equation}
\begin{equation}
	D=\frac{p+1+2x+8y}{16}, E=\frac{p-3-2x}{16}.\notag
\end{equation}
$y$ is two possible values, depending on the choice of the primitive element $\theta$ of $\mathbb{F}_{p}$.
\end{lemma}

 Define $\eta_{i}(4)=\sum_{t\in D_{i}}4^{t}(\mathrm{mod}\;4^{p}-1), i=0,1,2,3.$ Obviously, we have $\eta_{0}(4)+\eta_{1}(4)+\eta_{2}(4)+\eta_{3}(4)\equiv-1(\textup{mod}\:\frac{4^{p}-1}{3})$. And we have the following lemma.
\begin{lemma}\label{lem2}\cite{w12}
Let $k,l=0,1,2,3$. Then
\begin{equation}
	\eta_{l}(4)\cdot\eta_{l+k}(4)\equiv\sum_{f=0}^{3}(k,f)_{4}\eta_{f+l}(4)+\delta(\textup{mod}\;4^{p}-1),\notag
\end{equation}
where
\begin{equation}
	\delta=\begin{cases}
		\frac{p-1}{4},\qquad &\textup{if }\;p\equiv1(\textup{mod}8),k=0 \textup{ or } p\equiv5(\textup{mod}8),k=2, \\
		0,&\textup{otherwise }.
	\end{cases}\notag
\end{equation}
\end{lemma}

Let $H(4)=\eta_{0}(4)+\eta_{2}(4)-\eta_{1}(4)-\eta_{3}(4)$. Then we have the following lemma.
\begin{lemma}\label{lem3}\cite{w19}
	$H^{2}(4)\equiv p(\textup{mod}\;\frac{4^{p}-1}{3})$.
\end{lemma}

\subsection{Definition of sequences}
Let $\mathbb{Z}_{N}=\left\lbrace 0,1,\dots,N-1\right\rbrace $ denotes the ring of residue classes of modulo $N$ and $\mathbb{Z}_{N}^{*}$ denotes the set of the elements in $\mathbb{Z}_{N}$ that are coprime with $N$. The set $\textup{Supp}(s)=\left\lbrace t\in\mathbb{Z}_{N}: s(t)=1\right\rbrace $ is called the support set of a binary sequence $s$ with a period of $N$. The inverse Gray mapping $\phi$ is defined by $\phi[0,0]=0$, $\phi[0,1]=1$, $\phi[1,1]=2$, $\phi[1,0]=3$. This is a common method for constructing quaternary sequences.

In reference \cite{w14}, four classes of quaternary sequences with period $p$ were defined by using inverse Gray mapping. Due to the different support sets of component sequences, there are four sequences in each class. The first class of quaternary sequences are $u_{1}^{}=\phi[s_{1},s_{2}]$, $u_{2}^{}=\phi[s_{1},s_{3}]$, $u_{3}^{}=\phi[s_{2},s_{4}]$, $u_{4}^{}=\phi[s_{3},s_{4}]$; the second class of quaternary sequences are $u_{5}^{}=\phi[s_{1},s_{2}+1]$, $u_{6}^{}=\phi[s_{1},s_{3}+1]$, $u_{7}^{}=\phi[s_{2},s_{4}+1]$, $u_{8}^{}=\phi[s_{3},s_{4}+1]$; the third class of quaternary sequences are $u_{9}^{}=\phi[s_{1}+1,s_{2}]$, $u_{10}^{}=\phi[s_{1}+1,s_{3}]$, $u_{11}^{}=\phi[s_{2}+1,s_{4}]$, $u_{12}^{}=\phi[s_{3}+1,s_{4}]$; the fourth class of quaternary sequences are $u_{13}^{}=\phi[s_{1}+1,s_{2}+1]$, $u_{14}^{}=\phi[s_{1}+1,s_{3}+1]$, $u_{15}^{}=\phi[s_{2}+1,s_{4}+1]$, $u_{16}^{}=\phi[s_{3}+1,s_{4}+1]$. Here $s_{1}$, $s_{2}$, $s_{3}$, $s_{4}$ are four binary cyclotomic sequences with support sets $D_{0}\cup D_{1}$, $D_{0}\cup D_{3}$, $D_{1}\cup D_{2}$ and $D_{2}\cup D_{3}$, respectively.

For a quaternary sequence $u$ with period $N$, its generating polynomial $U(x)$ is defined as $\sum_{t=0}^{N-1}u(t)x^{t}$. The 4-adic complexity, denoted by $\Phi_{4}(u)$, is defined by
\begin{equation}
	\lfloor \mathrm{log}_{4}(\dfrac{4^{N}-1}{\mathrm{gcd}(4^{N}-1,U(4))}+1)\rfloor,\notag
\end{equation}
where $\lfloor m\rfloor$ is the greatest integer that is less than or equal to $m$ \cite{w7}.

In the rest of the paper, $(\frac{\cdot}{\cdot})$ denotes the Legendre symbol.

\section{4-Adic complexity of quaternary sequences}
In this section, we will estimate the 4-adic complexity of four classes of quaternary sequences in reference \cite{w14}.
\begin{lemma}\label{lem5}
	Let $p=x^{2}+4y^{2}$, where $x,y\in\mathbb{Z}$ and $x\equiv1(\textup{mod}\;4)$. When $p\equiv5(\textup{mod}\;8)$, then
\begin{itemize}
\item [\rm{(1)}]
     $4(3+\eta_{1}(4)+2\eta_{2}(4)+3\eta_{3}(4))(3+\eta_{3}(4)+2\eta_{0}(4)+3\eta_{1}(4))\equiv-2(3+4y)H(4)+5p+9(\textup{mod}\;\frac{4^{p}-1}{3});$
\item [\rm{(2)}]
       $4(1+\eta_{2}(4)+2\eta_{1}(4)+3\eta_{0}(4))(1+\eta_{0}(4)+2\eta_{3}(4)+3\eta_{2}(4))\equiv-2(1+4y)H(4)+5p+1(\textup{mod}\;\frac{4^{p}-1}{3});$
\item [\rm{(3)}]
      $4(\eta_{3}(4)+2\eta_{0}(4)+3\eta_{1}(4))(\eta_{1}(4)+2\eta_{2}(4)+3\eta_{3}(4))\equiv2(3-4y)H(4)+5p+9(\textup{mod}\;\frac{4^{p}-1}{3});$
\item [\rm{(4)}] $4(\eta_{2}(4)+2\eta_{3}(4)+3\eta_{0}(4))(\eta_{0}(4)+2\eta_{1}(4)+3\eta_{2}(4))\equiv-2(3-4y)H(4)+5p+9(\textup{mod}\;\frac{4^{p}-1}{3});$
\item [\rm{(5)}]
    $4(\eta_{3}(4)+2\eta_{2}(4)+3\eta_{1}(4))(\eta_{1}(4)+2\eta_{0}(4)+3\eta_{3}(4))\equiv2(3+4y)H(4)+5p+9(\textup{mod}\;\frac{4^{p}-1}{3});$
\item [\rm{(6)}]
  $4(1+\eta_{3}(4)+2\eta_{0}(4)+3\eta_{1}(4))(1+\eta_{1}(4)+2\eta_{2}(4)+3\eta_{3}(4))\equiv2(1-4y)H(4)+5p+1(\textup{mod}\;\frac{4^{p}-1}{3});$
\item [\rm{(7)}]
  $4(1+\eta_{1}(4)+2\eta_{0}(4)+3\eta_{3}(4))(1+\eta_{3}(4)+2\eta_{2}(4)+3\eta_{1}(4))\equiv2(1+4y)H(4)+5p+1(\textup{mod}\;\frac{4^{p}-1}{3});$
\item [\rm{(8)}]
  $4(2+\eta_{1}(4)+2\eta_{0}(4)+3\eta_{3}(4))(2+\eta_{3}(4)+2\eta_{2}(4)+3\eta_{1}(4))\equiv-2(1-4y)H(4)+5p+1(\textup{mod}\;\frac{4^{p}-1}{3}).$
\end{itemize}
\end{lemma}
\begin{proof} The proof of this lemma can be completed by using Lemmas \ref{lem1} and \ref{lem2}. Since the proofs are similar, we only prove (1). Firstly, it is easy to verify that
\begin{align*}
&4(3+\eta_{1}(4)+2\eta_{2}(4)+3\eta_{3}(4))(3+\eta_{3}(4)+2\eta_{0}(4)+3\eta_{1}(4))\\&=36+24\eta_{0}(4)+48\eta_{1}(4)+24\eta_{2}(4)+48\eta_{3}(4)+8\eta_{1}(4)\eta_{0}(4)+12\eta_{1}(4)\eta_{1}(4)\\&+4\eta_{1}(4)\eta_{3}(4)+16\eta_{2}(4)\eta_{0}(4)+24\eta_{2}(4)\eta_{1}(4)+8\eta_{2}(4)\eta_{3}(4)+24\eta_{3}(4)\eta_{0}(4)\\&+36\eta_{3}(4)\eta_{1}(4)+12\eta_{3}(4)\eta_{3}(4).
\end{align*}
According to  Lemmas \ref{lem1} and \ref{lem2}, we have
\begin{equation*}
	\eta_{1}(4)\eta_{0}(4)=E\eta_{0}(4)+E\eta_{1}(4)+D\eta_{2}(4)+B\eta_{3}(4)(\textup{mod}\;4^{p}-1);
\end{equation*}
\begin{equation*}
	\eta_{1}(4)\eta_{1}(4)=D\eta_{0}(4)+A\eta_{1}(4)+B\eta_{2}(4)+C\eta_{3}(4)(\textup{mod}\;4^{p}-1);
\end{equation*}
\begin{equation*}
	\eta_{1}(4)\eta_{3}(4)=E\eta_{0}(4)+A\eta_{1}(4)+E\eta_{2}(4)+A\eta_{3}(4)+\frac{p-1}{4}(\textup{mod}\;4^{p}-1);
\end{equation*}
\begin{equation*}
	\eta_{2}(4)\eta_{0}(4)=A\eta_{0}(4)+E\eta_{1}(4)+A\eta_{2}(4)+E\eta_{3}(4)+\frac{p-1}{4}(\textup{mod}\;4^{p}-1);
\end{equation*}
\begin{equation*}
	\eta_{2}(4)\eta_{1}(4)=B\eta_{0}(4)+E\eta_{1}(4)+E\eta_{2}(4)+D\eta_{3}(4)(\textup{mod}\;4^{p}-1);
\end{equation*}
\begin{equation*}
	\eta_{2}(4)\eta_{3}(4)=D\eta_{0}(4)+B\eta_{1}(4)+E\eta_{2}(4)+E\eta_{3}(4)(\textup{mod}\;4^{p}-1);
\end{equation*}
\begin{equation*}
	\eta_{3}(4)\eta_{0}(4)=E\eta_{0}(4)+D\eta_{1}(4)+B\eta_{2}(4)+E\eta_{3}(4)(\textup{mod}\;4^{p}-1);
\end{equation*}
\begin{equation*}
	\eta_{3}(4)\eta_{1}(4)=E\eta_{0}(4)+A\eta_{1}(4)+E\eta_{2}(4)+A\eta_{3}(4)+\frac{p-1}{4}(\textup{mod}\;4^{p}-1);
\end{equation*}
\begin{equation*}
	\eta_{3}(4)\eta_{3}(4)=B\eta_{0}(4)+C\eta_{1}(4)+D\eta_{2}(4)+A\eta_{3}(4)(\textup{mod}\;4^{p}-1),
\end{equation*}
then we get
\begin{align*}
	&4(3+\eta_{1}(4)+2\eta_{2}(4)+3\eta_{3}(4))(3+\eta_{3}(4)+2\eta_{0}(4)+3\eta_{1}(4))\\&=(16A+36B+20D+72E+24)\eta_{0}(4)\\&+(52A+8B+12C+24D+48E+48)\eta_{1}(4)\\&+(16A+36B+20D+72E+24)\eta_{2}(4)\\&+(52A+8B+12C+24D+48E+48)\eta_{3}(4)+14p+22(\textup{mod}\;4^{p}-1).
\end{align*}
Again by Lemma \ref{lem1}, we get
\begin{align*}
	&4(3+\eta_{1}(4)+2\eta_{2}(4)+3\eta_{3}(4))(3+\eta_{3}(4)+2\eta_{0}(4)+3\eta_{1}(4))\\&=(9p+7-8y)\eta_{0}(4)+(9p+19+8y)\eta_{1}(4)+(9p+7-8y)\eta_{2}(4)\\&+(9p+19+8y)\eta_{3}(4)+14p+22(\textup{mod}\;4^{p}-1)\\&=(9p+13)(\eta_{0}(4)+\eta_{1}(4)+\eta_{2}(4)+\eta_{3}(4))\\&-(8y+6)(\eta_{0}(4)-\eta_{1}(4)+\eta_{2}(4)-\eta_{3}(4))+14p+22(\textup{mod}\;4^{p}-1)\\&\equiv-2(3+4y)H(4)+5p+9(\textup{mod}\;\frac{4^{p}-1}{3}),
\end{align*}
which completes the proof.
\end{proof}	

\begin{lemma}\label{lem6}
	Let $U_{3}(4)=\eta_{2}(4)+2\eta_{3}(4)+3\eta_{0}(4)(\textup{mod}\;4^{p}-1)$. If $d_{3}>3$ is a prime divisor of $gcd(U_{3}(4),\frac{4^{p}-1}{3})$. Then either $d_{3}=11$ or $d_{3}=1+2p$ where $3p=4(3-4y)^{2}+83$.
\end{lemma}

\begin{proof}
	Let $p\equiv5\pmod8$ and $d_{3}>3$ be a prime divisor of $\textup{gcd}(U_{3}(4),\frac{4^{p}-1}{3})$. According to Lemma \ref{lem5} $\textup{(4)}$, we have $d_{3}|-2(3-4y)H(4)+5p+9$, that is $2(3-4y)H(4)\equiv5p+9(\textup{mod}\;d_{3})$. Then by Lemma \ref{lem3} we have $4(3-4y)^{2}p\equiv25p^{2}+90p+81(\textup{mod}\;d_{3})$. From $4^{p}\equiv1(\textup{mod}\;d_{3})$ and $d_{3}\not=3$, we have the order of 4 modulo $d_{3}$ is $p$. Since $d_{3}$ is odd prime, by Euler's Theorem we have $4^{d_{3}-1}\equiv1(\textup{mod}\;d_{3})$. Thus $p|d_{3}-1$. Then $d_{3}=1+2fp$, where $f\in\mathbb{N_{+}}$, that is $2fp\equiv-1(\textup{mod}\;d_{3})$. Hence we have $-4(3-4y)^{2}+25p+90-162f\equiv0(\textup{mod}\;d_{3})$ or $-4(3-4y)^{2}+25p+90-162f=(1+2h)(1+2fp)$. From the last congruence we get $3-2f\equiv1+2h+2f+4fh\pmod8$. Thus $h\equiv1\pmod2$ and $1+2h=3+8t$ where $t\in\mathbb{Z}$. So, we get
	\begin{equation}
		-4(3-4y)^{2}+25p+90-162f=(3+8t)(1+2fp)\tag{1}\label{code2}
	\end{equation}
	
	The discussion can be divided into three cases as follows.
	
	\noindent\textbf{(i) Case 1}: If $t<0$, since $25p+90-4(3-4y)^{2}\ge36y^{2}+96y+79=(6y+8)^{2}+15>0$, then by (\ref{code2}) we have $-162f<(3+8t)(1+2fp)$, that is $162f>(-3-8t)(1+2fp)\ge5+10fp$. Hence $p=5$ or $p=13$. Combining (\ref{code2}) we have, when $p=5$, $d_{3}=11$; when $p=13$, we obtain a contradiction.
	
	\noindent\textbf{(ii) Case 2}: If $t>0$, since $90-162f-4(3-4y)^{2}<0$, then by (\ref{code2}) we get $25p>(3+8t)(1+2fp)$. Hence $f=1$, $t=1$, that is $d_{3}=1+2p$ is a prime divisor of $\frac{4^{p}-1}{3}$ and $3p=4(3-4y)^{2}+83$.
	
	\noindent\textbf{(iii) Case 3}: If $t=0$, we have $-4(3-4y)^{2}+25p+90-162f=3(1+2fp)$. Since $90-4(3-4y)^{2}-162f-3<0$, it follows that $25p>6fp$. Hence, the possible values of $f$ are $1,2,3,4$ and $-y^{2}+p\equiv0\pmod3$, thus $p\equiv y^{2}\equiv1\pmod3$. In this case $f\not=1$ and $f\not=4$; if not, $1+2fp\equiv0\pmod3$, which is contradictory.
	
	\textbf{(iii-a)} Let $f=2$. Then $d_{3}=1+4p$ and $d_{3}\equiv5\pmod8$. Thus $(\frac{2}{d_{3}})=(-1)^{\frac{d_{3}^{2}-1}{8}}=-1$, that is 2 is not a quadratic residue of modulo $d_{3}$. Further $2^{2p}\not\equiv1(\textup{mod}\;d_{3})$. Then $4^{p}\not\equiv1(\textup{mod}\;d_{3})$. This contradicts with $d_{3}|4^{p}-1$, thus $d_{3}\not=1+4p$.
	
	\textbf{(iii-b)} Let $f=3$. Then $d_{3}=1+6p$, thus $-4(3-4y)^{2}+25p+90-486=3(1+6p)$, that is $-4(3-4y)^{2}=7(57-p)$. Hence $3-4y=7k$, where $k\in\mathbb{Z}$ and $-4\times7k^{2}+p-57=0$, then $p=4\times7k^{2}+57\equiv1(\textup{mod}\;7)$, also by $p\equiv5\pmod8$ and Chinese Remainder Theorem, we arrive at a contradiction.
	
	The proof is completed.
\end{proof}

\begin{theorem}\label{the1}
	Let $u_{3}=\phi[s_{2}, s_{4}]$. Then the 4-adic complexity of quaternary sequence $u_{3}$ satisfies $\Phi_{4}(u_{3})\ge\lfloor \textup{log}_{4}(\frac{4^{p}-1}{3(1+2p)}+1)\rfloor$.
\end{theorem}
	\begin{proof}
		To calculate the 4-adic complexity of $u_{3}=\phi[s_{2}, s_{4}]$, we need to consider $\textup{gcd}(U_{3}(4),4^{p}-1)$. Since
		\begin{align*}			U_{3}(4)&=\sum_{t=0}^{p-1}u_{3}(t)4^{t}(\textup{mod}\;4^{p}-1)\\&=\sum_{t=0}^{p-1}\phi[s_{2}(t),s_{4}(t)]4^{t}(\textup{mod}\;4^{p}-1)\\&=\sum_{\substack{t=0\\t\in D_{2}}}^{p-1}4^{t}+2\sum_{\substack{t=0\\t\in D_{3}}}^{p-1}4^{t}+3\sum_{\substack{t=0\\t\in D_{0} }}^{p-1}4^{t}(\textup{mod}\;4^{p}-1)\\&=\eta_{2}(4)+2\eta_{3}(4)+3\eta_{0}(4)(\textup{mod}\;4^{p}-1),
		\end{align*}
 we have $U_{3}(4)\equiv6\times\frac{p-1}{4}\equiv0\pmod3$ and $4^{p}\equiv1\pmod3$, hence 3 is a divisor of $\textup{gcd}(U_{3}(4),4^{p}-1)$. If $9|4^{p}-1$, then from $4^{3}\equiv1\pmod9$ and $4^{p}\equiv1\pmod9$, we obtain $p=3$, which contradicts with $p\equiv5\pmod8$. Thus 9 does not divide $4^{p}-1$, that is 9 does not divide $\textup{gcd}(U_{3}(4),4^{p}-1)$.
		
		Let $p\equiv5\pmod8$ and $r_{3}$ be a divisor of $\textup{gcd}(U_{3}(4),4^{p}-1)$ and $\textup{gcd}(r_{3},3)=1$. For $p=5$, by the definition of sequence, $U_{3}(4)=4^{4}+2\times4^{3}+3\times4=396$, thus $\textup{gcd}(U_{3}(4),4^{5}-1)=\textup{gcd}(396,1023)=33$ and $\Phi_{4}(u_{3})=2$. For $p>5$, according to Lemma \ref{lem6} we get $r_{3}$ is a power of $2p+1$ and $2p+1$ is a prime satisfying $3p=4(3-4y)^{2}+83$. If $(2p+1)^{2}|r_{3}$, by Lemma \ref{lem5} (4), we get $(2p+1)^{2}|-2(3-4y)H(4)+5p+9$, thus $4(3-4y)^{2}p\equiv25p^{2}+90p+81(\textup{mod}\;(2p+1)^{2})$ and by Lemma \ref{lem6} we have $3p=4(3-4y)^{2}+83$, that is $4(3-4y)^{2}=3p-83$. Thus $(3p-83)p-25p^{2}-90p-81\equiv0(\textup{mod}\;(2p+1)^{2})$. Then we have $-22p^{2}-173p-81=2(2p+1)^{2}m$, where $m\in\mathbb{Z}$. This is impossible.
		
		Therefore, the possible prime factor of $\textup{gcd}(U_{3}(4),4^{p}-1)$ is $2p+1$, and there is no repetition factor. And $\textup{gcd}(U_{3}(4),4^{p}-1)=3$ when $2p+1$ is not a prime divisor of $4^{p}-1$, which completes the proof.
	\end{proof}

\begin{theorem}\label{the2}
	Let $u_{1}=\phi[s_{1},s_{2}]$, $u_{9}=\phi[s_{1}+1,s_{2}]$, $u_{11}=\phi[s_{2}+1,s_{4}]$. Then for $i=1,9,11$, the 4-adic complexity of these quaternary sequences $u_{i}$ satisfies $\Phi_{4}(u_{i})\ge\lfloor \textup{log}_{4}(\frac{4^{p}-1}{3(1+2p)}+1)\rfloor$, where $1+2p=11$ or the $p$ in $1+2p$ satisfies $3p=4(3-4y)^{2}+83$.
\end{theorem}
\begin{proof}
	Since the proof is similar to Theorem \ref{the1}, we omit it.
\end{proof}

\begin{lemma}\label{lem7}
	Let $U_{5}(4)=1+\eta_{2}(4)+2\eta_{1}(4)+3\eta_{0}(4)(\textup{mod}\;4^{p}-1)$. If $d_{5}>3$ is a prime divisor of $gcd(U_{5}(4),\frac{4^{p}-1}{3})$. Then the possible values of $d_{5}$ are
\textup{(1)} $d_{5}=11$; \textup{(2)} $d_{5}=1+2p$ and $3p=4(1+4y)^{2}+3$; \textup{(3)} $d_{5}=1+6p$ and $7p=4(1+4y)^{2}-1$; \textup{(4)} $d_{5}=1+8p$ and $p=4(1+4y)^{2}+1$.
\end{lemma}
\begin{proof}
	Let $p\equiv5\pmod8$ and $d_{5}>3$ be a prime divisor of $\textup{gcd}(U_{5}(4),\frac{4^{p}-1}{3})$. According to Lemma \ref{lem5} $\textup{(2)}$, we have $d_{5}|-2(1+4y)H(4)+5p+1$, thus $5p+1\equiv2(1+4y)H(4)(\textup{mod}\;d_{5})$. Then by Lemma \ref{lem3} we have $4(1+4y)^{2}p\equiv25p^{2}+10p+1(\textup{mod}\;d_{5})$. From $4^{p}\equiv1(\textup{mod}\;d_{5})$ and $d_{5} \not=3$, we have the order of 4 modulo $d_{5}$ is $p$. Since $d_{5}$ is odd prime, by Euler's Theorem we have $4^{d_{5}-1}\equiv1(\textup{mod}\;d_{5})$. Thus $p|d_{5}-1$. Then $d_{5}=1+2fp$, where $f\in\mathbb{N_{+}}$, that is $2fp\equiv-1(\textup{mod}\;d_{5})$. Hence we have $-4(1+4y)^{2}+25p+10-2f\equiv0(\textup{mod}\;d_{5})$ or $-4(1+4y)^{2}+25p+10-2f=(1+2h)(1+2fp)$. From the last congruence we get $3-2f\equiv1+2f+2h+4fh(\textup{mod}\;8)$. Thus $h\equiv1(\textup{mod}\;2)$ and $1+2h=3+8t$ where $t\in\mathbb{Z}$. So, we get
	\begin{equation}
		-4(1+4y)^{2}+25p+10-2f=(3+8t)(1+2fp)\tag{2}\label{code3}
	\end{equation}

The discussion can be divided into three cases as follows.

\noindent\textbf{(i) Case 1}: If $t<0$, since $25p+10-4(1+4y)^{2}\ge36y^{2}-32y+31=(6y-\frac{8}{3})^{2}+\frac{215}{9}>0$,  then by (\ref{code3}) we have $-2f<(3+8t)(1+2fp)$, that is $2f>(-3-8t)(1+2fp)\ge5+10fp$. It is impossible.

\noindent\textbf{(ii) Case 2}: If $t>0$, since $10-2f-4(1+4y)^{2}<0$, then by (\ref{code3}) we get that $25p>(3+8t)(1+2fp)$. Hence $f=1,t=1$, that is $d_{5}=1+2p$ is a prime divisor of $\frac{4^{p}-1}{3}$ and $3p=4(1+4y)^{2}+3$.

\noindent\textbf{(iii) Case 3}: If $t=0$, we have $-4(1+4y)^{2}+25p+10-2f=3(1+2fp)$. Since $10-2f-4(1+4y)^{2}-3<0$, it follows that $25p>6fp$. Hence, the possible values of $f$ are $1,2,3,4$.

\textbf{(iii-a)} Let $f=1$. Then $d_{5}=1+2p$, thus $-4(1+4y)^{2}+25p+8=3(1+2p)$ or $19p=4(1+4y)^{2}-5$. In this case, only holds true when $x=1$ and $y=1$. It follows that $p=x^{2}+4y^{2}=5$ and $d_{5}=1+2p=11$.

\textbf{(iii-b)} Let $f=2$. Then $d_{5}=1+4p$ and $d_{5}\equiv5\pmod8$. 	Thus $\left( \frac{2}{d_{5}}\right)=-1$, that is 2 is not a quadratic residue of modulo $d_{5}$. Further $2^{2p}\not\equiv1(\textup{mod}\;d_{5})$. Then $4^{p}\not\equiv1(\textup{mod}\;d_{5})$. This contradicts with $d_{5}|4^{p}-1$, thus $d_{5}\not=1+4p$.

\textbf{(iii-c)} Let $f=3$. Then $d_{5}=1+6p$, thus $-4(1+4y)^{2}+25p+4=3(1+6p)$ or $7p=4(1+4y)^{2}-1$.

\textbf{(iii-d)} Let $f=4$. Then $d_{5}=1+8p$, thus $-4(1+4y)^{2}+25p+2=3(1+8p)$ or $p=4(1+4y)^{2}+1$.

The proof is completed.
\end{proof}

\begin{theorem}\label{the3}
	Let $u_{5}=\phi[s_{1},s_{2}+1]$. Then the 4-adic complexity of quaternary sequence $u_{5}$ satisfies $\Phi_{4}(u_{5})\ge\lfloor \textup{log}_{4}(\frac{4^{p}-1}{(1+2p)(1+6p)(1+8p)}+1)\rfloor$.
\end{theorem}
\begin{proof}To calculate the 4-adic complexity of $u_{5}=\phi[s_{1},s_{2}+1]$, we need to consider $\textup{gcd}(U_{5}(4),4^{p}-1)$. Since
		\begin{align*}
U_{5}(4)&=\sum_{t=0}^{p-1}u_{5}(t)4^{t}(\textup{mod}\;4^{p}-1)\\&=\sum_{t=0}^{p-1}\phi[s_{1}(t),s_{2}(t)+1]4^{t}(\textup{mod}\;4^{p}-1)\\&=\sum_{\substack{t=0\\t\in D_{2}\cup\left\lbrace 0\right\rbrace}}^{p-1}4^{t}+2\sum_{\substack{t=0\\t\in D_{1}}}^{p-1}4^{t}+3\sum_{\substack{t=0\\t\in D_{0}}}^{p-1}4^{t}(\textup{mod}\;4^{p}-1)\\&=1+\eta_{2}(4)+2\eta_{1}(4)+3\eta_{0}(4)(\textup{mod}\;4^{p}-1),
		\end{align*}
 we have $U_{5}(4)\equiv1+6\times \frac{p-1}{4}\equiv1\pmod3$, hence 3 does not divide $\textup{gcd}(U_{5}(4),4^{p}-1)$. Let $p\equiv5\pmod8$ and $r_{5}$ be a divisor of $\textup{gcd}(U_{5}(4),4^{p}-1)$ and $\textup{gcd}(r_{5},3)=1$. For $p=5$, by the definition of sequence, $U_{5}(4)=1+4^{4}+2\times4^{2}+3\times4=301$, thus $\textup{gcd}(U_{5}(4),4^{5}-1)=\textup{gcd}(301,1023)=1$ and $\Phi_{4}(u_{5})=5$. For $p>5$, according to Lemma \ref{lem7}, the possible prime factors of $r_{5}$ are $1+2p$, $1+6p$ and $1+8p$, and satisfy the conditions in Lemma \ref{lem7}.
	
	If $(1+2p)^{2}|r_{5}$, by Lemma \ref{lem5} (2), we get $(1+2p)^{2}|-2(1+4y)H(4)+5p+1$, then by Lemma \ref{lem3} we get  $4(1+4y)^{2}p\equiv25p^{2}+10p+1(\textup{mod}\;(1+2p)^{2})$ and by Lemma \ref{lem7} we have $3p=4(1+4y)^{2}+3$, that is $4(1+4y)^{2}=3p-3$. Thus $(3p-3)p\equiv25p^{2}+10p+1(\textup{mod}\;(1+2p)^{2})$. Then we have $-22p^{2}-13p-1=2(1+2p)^{2}m$, where $m\in\mathbb{Z}$. We obtain a contradiction.
	
	If $(1+6p)^{2}|r_{5}$, according to Lemma \ref{lem5} (2) we have $(1+6p)^{2}|-2(1+4y)H(4)+5p+1$, then by Lemma \ref{lem3} we get $4(1+4y)^{2}p\equiv25p^{2}+10p+1(\textup{mod}\;(1+6p)^{2})$ and by Lemma \ref{lem7} we have $7p=4(1+4y)^{2}-1$, that is $4(1+4y)^{2}=7p+1$. Thus $(7p+1)p\equiv25p^{2}+10p+1(\textup{mod}\;(1+6p)^{2})$. Then we get $-18p^{2}-9p-1=2(1+6p)^{2}m$, where $m\in\mathbb{Z}$. This will create a contradiction.
	
	If $(1+8p)^{2}|r_{5}$, by Lemma \ref{lem5} (2), we get $(1+8p)^{2}|-2(1+4y)H(4)+5p+1$, then by Lemma \ref{lem3} we get $4(1+4y)^{2}p\equiv25p^{2}+10p+1(\textup{mod}\;(1+8p)^{2})$ and by Lemma \ref{lem7} we have $p=4(1+4y)^{2}+1$, that is $4(1+4y)^{2}=p-1$. Thus $(p-1)p\equiv25p^{2}+10p+1(\textup{mod}\;(1+8p)^{2})$. Then we get $-24p^{2}-11p-1=2(1+8p)^{2}m$, where $m\in\mathbb{Z}$. This is impossible.
	
	Therefore, the possible prime factors of $\textup{gcd}(U_{5}(4),4^{p}-1)$ are $1+2p$, $1+6p$ and $1+8p$, and there are no repetition factors. And if $1+2p$, $1+6p$ and $1+8p$ are not prime factors of $4^{p}-1$, the prime factors of $\textup{gcd}(U_{5}(4),4^{p}-1)$ do not contain $1+2p$, $1+6p$ and $1+8p$, respectively, which completes the proof.
\end{proof}

\begin{theorem}\label{the4}
	Let $u_{7}=\phi[s_{2},s_{4}+1]$, $u_{13}=\phi[s_{1}+1,s_{2}+1]$, $u_{15}=\phi[s_{2}+1,s_{4}+1]$. Then for $i=7,13,15$, the 4-adic complexity of quaternary sequences $u_{i}$ satisfies $\Phi_{4}(u_{i})\ge\lfloor \textup{log}_{4}(\frac{4^{p}-1}{(1+2p)(1+6p)(1+8p)}+1)\rfloor$ , where either $p=5$ or the $p$ in $1+2p$ satisfies $3p=4(1+4y)^{2}+3$, the $p$ in $1+6p$ satisfies $7p=4(1+4y)^{2}-1$ and the $p$ in $1+8p$ satisfies $p=4(1+4y)^{2}+1$.
\end{theorem}
\begin{proof}Since the proof is similar to Theorem \ref{the3}, we omit it.
\end{proof}

\begin{lemma}\label{lem8}
	Let $U_{10}(4)=3+\eta_{1}(4)+2\eta_{2}(4)+3\eta_{3}(4)(\textup{mod}\;4^{p}-1)$. If $d_{10}>3$ is a prime divisor of $\textup{gcd}(U_{10}(4),\frac{4^{p}-1}{3})$. Then the possible values of $d_{10}$ are $\textup{(1) } d_{10}=11$; $\textup{(2) } d_{10}=1+2p$ and $3p=4(3+4y)^{2}+83$.
\end{lemma}
\begin{proof}
	Let $p\equiv5\pmod8$ and $d_{10}>3$ be a prime divisor of $\textup{gcd}(U_{10}(4),\frac{4^{p}-1}{3})$. According to Lemma \ref{lem5} $\textup{(1)}$, we have $d_{10}|-2(3+4y)H(4)+5p+9$, that is $2(3+4y)H(4)\equiv5p+9(\textup{mod}\;d_{10})$. Then by Lemma \ref{lem3} we have $4(3+4y)^{2}p\equiv25p^{2}+90p+81(\textup{mod}\;d_{10})$. From $4^{p}\equiv1(\textup{mod}\;d_{10})$ and $d_{10}\not=3$, we have the order of 4 modulo $d_{10}$ is $p$. Since $d_{10}$ is odd prime, by Euler's Theorem we have $4^{d_{10}-1}\equiv1(\textup{mod}\;d_{10})$. Thus $p|d_{10}-1$. Then $d_{10}=1+2fp$, where $f\in\mathbb{N_{+}}$, that is $2fp\equiv-1(\textup{mod}\;d_{10})$. Hence we have $-4(3+4y)^{2}+25p+90-162f\equiv0(\textup{mod}\;d_{10})$ or $-4(3+4y)^{2}+25p+90-162f=(1+2h)(1+2fp)$. From the last congruence we get $3-2f\equiv1+2h+2f+4fh\pmod8$. Thus $h\equiv1\pmod2$ and $1+2h=3+8t$ where $t\in\mathbb{Z}$. So, we get
	\begin{equation}
		-4(3+4y)^{2}+25p+90-162f=(3+8t)(1+2fp)\tag{3}\label{code4}
	\end{equation}

The discussion can be divided into three cases as follows.

\noindent\textbf{(i) Case 1}: If $t<0$, since $25p+90-4(3+4y)^{2}\ge36y^{2}-96y+79=(6y-8)^{2}+15>0$, then by (\ref{code4}) we have $-162f<(3+8t)(1+2fp)$, that is $162f>(-3-8t)(1+2fp)\ge5+10fp$. Hence $p=5$ or $p=13$. Combining (\ref{code4}) we have, when $p=5$, $d_{10}=11$; when $p=13$, we obtain a contradiction.

\noindent\textbf{(ii) Case 2} If $t>0$, since $90-162f-4(3+4y)^{2}<0$, then by (\ref{code4}) we get $25p>(3+8t)(1+2fp)$. Hence $f=1$, $t=1$, that is $d_{10}=1+2p$ is a prime divisor of $\frac{4^{p}-1}{3}$ and $3p=4(3+4y)^{2}+83$.

\noindent\textbf{(iii) Case 3} If $t=0$, we have $-4(3+4y)^{2}+25p+90-162f=3(1+2fp)$. Since $90-4(3+4y)^{2}-162f-3<0$, it follows that $25p>6fp$. Hence, the possible values of $f$ are $1,2,3,4$ and $-y^{2}+p\equiv0\pmod3$, thus $p\equiv y^{2}\equiv1\pmod3$. In this case $f\not=1$ and $f\not=4$; if not, $1+2fp\equiv0\pmod3$, which is contradictory.

\textbf{(iii-a)} Let $f=2$. Then $d_{10}=1+4p$ and $d_{10}\equiv5\pmod8$. Thus $(\frac{2}{d_{10}})=(-1)^{\frac{d_{10}^{2}-1}{8}}=-1$, that is 2 is not a quadratic residue of modulo $d_{10}$. Further $2^{2p}\not\equiv1(\textup{mod}\;d_{10})$. Then $4^{p}\not\equiv1(\textup{mod}\;d_{10})$. This contradicts with $d_{10}|4^{p}-1$, thus $d_{10}\not=1+4p$.

\textbf{(iii-b)} Let $f=3$. Then $d_{10}=1+6p$, thus $-4(3+4y)^{2}+25p+90-486=3(1+6p)$, that is $-4(3+4y)^{2}=7(57-p)$. Hence $3+4y=7k$, where $k\in\mathbb{Z}$ and $-4\times7k^{2}+p-57=0$, then $p=4\times7k^{2}+57\equiv1(\textup{mod}\;7)$, also by $p\equiv5\pmod8$ and Chinese Remainder Theorem, we arrive at a contradiction.

The proof is completed.
\end{proof}

\begin{theorem}\label{the5}
	Let $u_{10}=\phi[s_{1}+1,s_{3}]$. Then the 4-adic complexity of quaternary sequence $u_{10}$ satisfies $\Phi_{4}(u_{10})\ge\lfloor \textup{log}_{4}(\frac{4^{p}-1}{3(1+2p)}+1)\rfloor$.
\end{theorem}
\begin{proof}
	To calculate the 4-adic complexity of $u_{10}=\phi[s_{1}+1,s_{3}]$, we need to consider $\textup{gcd}(U_{10}(4),4^{p}-1)$. Since
	\begin{align*}		U_{10}(4)&=\sum_{t=0}^{p-1}u_{10}(t)4^{t}(\textup{mod}\;4^{p}-1)\\&=\sum_{t=0}^{p-1}\phi[s_{1}(t)+1,s_{3}(t)]4^{t}(\textup{mod}\;4^{p}-1)\\&=\sum_{\substack{t=0\\t\in D_{1}}}^{p-1}4^{t}+2\sum_{\substack{t=0\\t\in D_{2}}}^{p-1}4^{t}+3\sum_{\substack{t=0\\t\in D_{3}\cup\left\lbrace 0\right\rbrace }}^{p-1}4^{t}(\textup{mod}\;4^{p}-1)\\&=3+\eta_{1}(4)+2\eta_{2}(4)+3\eta_{3}(4)(\textup{mod}\;4^{p}-1),
	\end{align*}
then we have $U_{10}(4)\equiv3+6\times\frac{p-1}{4}\equiv0\pmod3$ and $4^{p}\equiv1\pmod3$, hence $3|\textup{gcd}(U_{10}(4),4^{p}-1)$. If $9|4^{p}-1$, then from $4^{3}\equiv1\pmod9$ and $4^{p}\equiv1\pmod9$, we obtain $p=3$, which contradicts with $p\equiv5\pmod8$. Thus 9 does not divide $4^{p}-1$, that is 9 does not divide $\textup{gcd}(U_{10}(4),4^{p}-1)$.

 Let $p\equiv5\pmod8$ and $r_{10}$ be a divisor of $\textup{gcd}(U_{10}(4),4^{p}-1)$ and $\textup{gcd}(r_{10},3)=1$. For $p=5$, by the definition of sequence, $U_{10}(4)=3+4^{2}+2\times4^{4}+3\times4^{3}=723$, thus $\textup{gcd}(U_{10}(4),4^{5}-1)=\textup{gcd}(723,1023)=3$ and $\Phi_{4}(u_{10})=4$. For $p>5$, according to Lemma \ref{lem8} we get $r_{10}$ is a power of $2p+1$ and $2p+1$ is a prime that satisfying $3p=4(3+4y)^{2}+83$. If $(2p+1)^{2}|r_{10}$, by Lemma \ref{lem5} (1), we get $(2p+1)^{2}|-2(3+4y)H(4)+5p+9$, thus $4(3+4y)^{2}p\equiv25p^{2}+90p+81(\textup{mod}\;(2p+1)^{2})$ and by Lemma \ref{lem8} we have $3p=4(3+4y)^{2}+83$, that is $4(3+4y)^{2}=3p-83$. Thus $(3p-83)p-25p^{2}-90p-81\equiv0(\textup{mod}\;(2p+1)^{2})$. Then we have $-22p^{2}-173p-81=2(2p+1)^{2}m$, where $m\in\mathbb{Z}$. This is impossible.

Therefore, the possible prime factor of $\textup{gcd}(U_{10}(4),4^{p}-1)$ is $2p+1$, and there is no repetition factor. And $\textup{gcd}(U_{10}(4),4^{p}-1)=3$ when $2p+1$ is not a prime divisor of $4^{p}-1$, which completes the proof.
\end{proof}

\begin{theorem}\label{the6}
Let $u_{2}=\phi[s_{1},s_{3}]$, $u_{4}=\phi[s_{3},s_{4}]$, $u_{12}=\phi[s_{3}+1,s_{4}]$. Then the 4-adic complexity of these quaternary sequences satisfies $\Phi_{4}(u_{i})\ge\lfloor \textup{log}_{4}(\frac{4^{p}-1}{3(1+2p)}+1)\rfloor$, $i=2,4,12$, where $1+2p=11$ or the $p$ in $1+2p$ satisfies $3p=4(3+4y)^{2}+83$.
\begin{proof}
	Since the proof is similar to Theorem \ref{the5}, we omit it.
\end{proof}
\end{theorem}

\begin{lemma}\label{lem9}
	Let $U_{16}(4)=2+\eta_{1}(4)+2\eta_{0}(4)+3\eta_{3}(4)(\textup{mod}\;4^{p}-1)$. If $d_{16}>3$ is a prime divisor of $gcd(U_{16}(4),\frac{4^{p}-1}{3})$. Then the possible values of $d_{16}$ are
	\textup{(1)} $d_{16}=11$; \textup{(2)} $d_{16}=1+2p$ and $3p=4(1-4y)^{2}+3$; \textup{(3)} $d_{16}=1+6p$ and $7p=4(1-4y)^{2}-1$; \textup{(4)} $d_{16}=1+8p$ and $p=4(1-4y)^{2}+1$.
\end{lemma}
	\begin{proof}
		Let $p\equiv5\pmod8$ and $d_{16}>3$ be a prime divisor of $\textup{gcd}(U_{16}(4),\frac{4^{p}-1}{3})$. According to Lemma \ref{lem5} $\textup{(8)}$, we have $d_{16}|-2(1-4y)H(4)+5p+1$, thus $5p+1\equiv2(1-4y)H(4)(\textup{mod}\;d_{16})$. Then by Lemma \ref{lem3} we have $4(1-4y)^{2}p\equiv25p^{2}+10p+1(\textup{mod}\;d_{16})$. From $4^{p}\equiv1(\textup{mod}\;d_{16})$ and $d_{16} \not=3$, we have the order of 4 modulo $d_{16}$ is $p$. Since $d_{16}$ is odd prime, by Euler's Theorem we have $4^{d_{16}-1}\equiv1(\textup{mod}\;d_{16})$. Thus $p|d_{16}-1$. Then $d_{16}=1+2fp$, where $f\in\mathbb{N_{+}}$, that is $2fp\equiv-1(\textup{mod}\;d_{16})$. Hence we have $-4(1-4y)^{2}+25p+10-2f\equiv0(\textup{mod}\;d_{16})$ or $-4(1-4y)^{2}+25p+10-2f=(1+2h)(1+2fp)$. From the last congruence we get $3-2f\equiv1+2f+2h+4fh(\textup{mod}\;8)$. Thus $h\equiv1(\textup{mod}\;2)$ and $1+2h=3+8t$ where $t\in\mathbb{Z}$. So, we get
		\begin{equation}
			-4(1-4y)^{2}+25p+10-2f=(3+8t)(1+2fp)\tag{4}\label{code5}
		\end{equation}
		
		The discussion can be divided into three cases as follows.
		
		\noindent\textbf{(i) Case 1}: If $t<0$, since $25p+10-4(1-4y)^{2}\ge36y^{2}+32y+31=(6y+\frac{8}{3})^{2}+\frac{215}{9}>0$,  then by (\ref{code5}) we have $-2f<(3+8t)(1+2fp)$, that is $2f>(-3-8t)(1+2fp)\ge5+10fp$. It is impossible.
		
		\noindent\textbf{(ii) Case 2}: If $t>0$, since $10-2f-4(1-4y)^{2}<0$, then by (\ref{code5}) we get that $25p>(3+8t)(1+2fp)$. Hence $f=1,t=1$, that is $d_{16}=1+2p$ is a prime divisor of $\frac{4^{p}-1}{3}$ and $3p=4(1-4y)^{2}+3$.
		
		\noindent\textbf{(iii) Case 3}: If $t=0$, here $-4(1-4y)^{2}+25p+10-2f=3(1+2fp)$. Since $10-2f-4(1-4y)^{2}-3<0$, it follows that $25p>6fp$. Hence, the possible values of $f$ are $1,2,3,4$.
		
		\textbf{(iii-a)} Let $f=1$. Then $d_{16}=1+2p$, thus $-4(1-4y)^{2}+25p+8=3(1+2p)$ or $19p=4(1-4y)^{2}-5$. In this case, only holds true when $x=1$ and $y=-1$. It follows that $p=x^{2}+4y^{2}=5$ and $d_{16}=1+2p=11$.
		
		\textbf{(iii-b)} Let $f=2$. Then $d_{16}=1+4p$ and $d_{16}\equiv5\pmod8$. 	Thus $\left( \frac{2}{d_{16}}\right)=-1$, that is 2 is not a quadratic residue of modulo $d_{16}$. Further $2^{2p}\not\equiv1(\textup{mod}\;d_{16})$. Then $4^{p}\not\equiv1(\textup{mod}\;d_{16})$. This contradicts with $d_{16}|4^{p}-1$, thus $d_{16}\not=1+4p$.
		
		\textbf{(iii-c)} Let $f=3$. Then $d_{16}=1+6p$, thus $-4(1-4y)^{2}+25p+4=3(1+6p)$ or $7p=4(1-4y)^{2}-1$.
		
		\textbf{(iii-d)} Let $f=4$. Then $d_{16}=1+8p$, thus $-4(1-4y)^{2}+25p+2=3(1+8p)$ or $p=4(1-4y)^{2}+1$.
		
		The proof is completed.
	\end{proof}

\begin{theorem}\label{the7}
	Let $u_{16}=\phi[s_{3}+1,s_{4}+1]$. Then the 4-adic complexity of quaternary sequence $u_{16}$ satisfies $\Phi_{4}(u_{16})\ge\lfloor \textup{log}_{4}(\frac{4^{p}-1}{(1+2p)(1+6p)(1+8p)}+1)\rfloor$.
\end{theorem}
	\begin{proof}
		To calculate the 4-adic complexity of $u_{16}=\phi[s_{3}+1,s_{4}+1]$, we need to consider $\textup{gcd}(U_{16}(4),4^{p}-1)$. Since
		\begin{align*}			U_{16}(4)&=\sum_{t=0}^{p-1}u_{16}(t)4^{t}(\textup{mod}\;4^{p}-1)\\&=\sum_{t=0}^{p-1}\phi[s_{3}(t)+1,s_{4}(t)+1]4^{t}(\textup{mod}\;4^{p}-1)\\&=\sum_{\substack{t=0\\t\in D_{1}}}^{p-1}4^{t}+2\sum_{\substack{t=0\\t\in D_{0}\cup\left\lbrace 0\right\rbrace}}^{p-1}4^{t}+3\sum_{\substack{t=0\\t\in D_{3}}}^{p-1}4^{t}(\textup{mod}\;4^{p}-1)\\&=2+\eta_{1}(4)+2\eta_{0}(4)+3\eta_{3}(4)(\textup{mod}\;4^{p}-1),
		\end{align*}
then $U_{16}(4)\equiv2+6\times \frac{p-1}{4}\equiv2\pmod3$, hence 3 does not divide $\textup{gcd}(U_{16}(4),4^{p}-1)$. Let $p\equiv5\pmod8$ and $r_{16}$ be a divisor of $\textup{gcd}(U_{16}(4),4^{p}-1)$ and $\textup{gcd}(r_{16},3)=1$. For $p=5$, by the definition of sequence, $U_{16}(4)=2+4^{2}+2\times4+3\times4^{3}=218$, thus $\textup{gcd}(U_{16}(4),4^{5}-1)=\textup{gcd}(218,1023)=1$ and $\Phi_{4}(u_{16})=5$. For $p>5$, according to Lemma \ref{lem9}, the possible prime factors of $r_{16}$ are $1+2p$, $1+6p$ and $1+8p$, and satisfy the conditions in Lemma \ref{lem9}.
		
		If $(1+2p)^{2}|r_{16}$, by Lemma \ref{lem5} (8), we get $(1+2p)^{2}|-2(1-4y)H(4)+5p+1$, then by Lemma \ref{lem3} we get  $4(1-4y)^{2}p\equiv25p^{2}+10p+1(\textup{mod}\;(1+2p)^{2})$ and by Lemma \ref{lem9} we have $3p=4(1-4y)^{2}+3$, that is $4(1-4y)^{2}=3p-3$. Thus $(3p-3)p\equiv25p^{2}+10p+1(\textup{mod}\;(1+2p)^{2})$. Then we have $-22p^{2}-13p-1=2(1+2p)^{2}m$, where $m\in\mathbb{Z}$. We obtain a contradiction.
		
		If $(1+6p)^{2}|r_{16}$, according to Lemma \ref{lem5} (8) we have $(1+6p)^{2}|-2(1-4y)H(4)+5p+1$, then by Lemma \ref{lem3} we get $4(1-4y)^{2}p\equiv25p^{2}+10p+1(\textup{mod}\;(1+6p)^{2})$ and by Lemma \ref{lem9} we have $7p=4(1-4y)^{2}-1$, that is $4(1-4y)^{2}=7p+1$. Thus $(7p+1)p\equiv25p^{2}+10p+1(\textup{mod}\;(1+6p)^{2})$. Then we get $-18p^{2}-9p-1=2(1+6p)^{2}m$, where $m\in\mathbb{Z}$. This will create a contradiction.
		
		If $(1+8p)^{2}|r_{16}$, by Lemma \ref{lem5} (8), we get $(1+8p)^{2}|-2(1-4y)H(4)+5p+1$, then by Lemma \ref{lem3} we get $4(1-4y)^{2}p\equiv25p^{2}+10p+1(\textup{mod}\;(1+8p)^{2})$ and by Lemma \ref{lem9} we have $p=4(1-4y)^{2}+1$, that is $4(1-4y)^{2}=p-1$. Thus $(p-1)p\equiv25p^{2}+10p+1(\textup{mod}\;(1+8p)^{2})$. Then we get $-24p^{2}-11p-1=2(1+8p)^{2}m$, where $m\in\mathbb{Z}$. This is impossible.
		
		Therefore, the possible prime factors of $\textup{gcd}(U_{16}(4),4^{p}-1)$ are $1+2p$, $1+6p$ and $1+8p$, and there are no repetition factors. And if $1+2p$, $1+6p$ and $1+8p$ are not prime factors of $4^{p}-1$, the prime factors of $\textup{gcd}(U_{16}(4),4^{p}-1)$ do not contain $1+2p$, $1+6p$ and $1+8p$, respectively, which completes the proof.
	\end{proof}

\begin{theorem}\label{the8}
	Let $u_{6}=\phi[s_{1},s_{3}+1]$, $u_{8}=\phi[s_{3},s_{4}+1]$, $u_{14}=\phi[s_{1}+1,s_{3}+1]$. Then the 4-adic complexity of these quaternary sequences satisfies $\Phi_{4}(u_{i})\ge\lfloor \textup{log}_{4}(\frac{4^{p}-1}{(1+2p)(1+6p)(1+8p)}+1)\rfloor$, $i=6,8,14$, where $1+2p=11$ or the $p$ in $1+2p$ satisfies $3p=4(1-4y)^{2}+3$, the $p$ in $1+6p$ satisfies $7p=4(1-4y)^{2}-1$, the $p$ in $1+8p$ satisfies $p=4(1-4y)^{2}+1$.
\end{theorem}
	\begin{proof}
		Since the proof is similar to Theorem \ref{the7}, we omit it.
	\end{proof}

\section{Conclusions}
It is interesting to study the 4-adic complexity of sequences with good correlation generated by FCSRs. In this paper, we estimate the 4-adic complexity of several classes of  quaternary sequences in reference \cite{w14}, which have low autocorrelation and high linear complexity. It turns out these sequences have high 4-adic complexity, which is sufficient safe to resist the attack of the RAA.


\bigskip
\begin{thebibliography}{99}
\baselineskip=11pt	

\bibitem{w1}
F. Adachi, D. Garg, S. Takaoka and K. Takeda, Broadband CDMA techniques, \newblock \emph{IEEE Wireless Communications}, \textbf{12} (2005), 8--18.


\bibitem{w2}
\newblock S. W. Golomb and G. Gong,
\newblock {Signal Design for Good Correlation: For Wireless Communication, Cryptography, and Radar},
\newblock Cambridge University Press, Cambridge, 2005.

\bibitem{w3}
\newblock S. W. Golomb,
\newblock {Shift Register Sequences},
\newblock CA:Aegean Park Press, 1967.

\bibitem{w4}
\newblock M. Goresky and A. Klapper,
\newblock {Algebraic Shift Register Sequences},
\newblock Cambridge University Press, Cambridgeshire, 2012.

\bibitem{w5}
\newblock A. Klapper and M. Goresky,
\newblock {Feedback shift registers, 2-adic span, and combiners with memory},
\newblock \emph{Journal of Cryptology}, \textbf{10} (1997), 111--147.

\bibitem{w6} 
\newblock A. Klapper and J. Xu,
\newblock {Algebraic feedback shift registers},
\newblock \emph{Theoretical Computer Science}, \textbf{226} (1999), 61--92.

\bibitem{w7} 
\newblock M. Yang, S. Qiang, K. Feng and D. Lin,
\newblock {On the 4-adic complexity of quaternary sequences of period $2p$ with ideal autocorrelation},
\newblock \emph{2021 IEEE International Symposium on Information Theory (ISIT)}, (2021), 1812--1816.

\bibitem{w8} 
\newblock M. Yang, S. Qiang, X. Jing, K. Feng and D. Lin,
\newblock {The 4-adic complexity of quaternary sequences of even period with ideal autocorrelation},
\newblock \emph{2022 IEEE International Symposium on Information Theory (ISIT)}, (2022), 528--531.

\bibitem{w9} 
\newblock X. Jing, Z. Xu, M. Yang and K. Feng,
\newblock {The 4-adic complexity of interleaved quaternary sequences of even length with optimal autocorrelation},
\newblock \emph{arXiv preprint arXiv:2209.10279} (2022).

\bibitem{w10} 
\newblock V. Edemskiy,
\newblock {Symmetric 4-adic complexity of quaternary sequences with low autocorrelation and period $pq$},
\newblock \emph{Advances in Mathematics of Communications}, (2023), 0--0.

\bibitem{w11} 
\newblock V, Edemskiy and S. Koltsova,
\newblock {Symmetric 4-adic complexity of quaternary sequences of length $pq$ with low autocorrelation},
\newblock \emph{2023 IEEE Information Theory Workshop (ITW)}, (2023), 76--80.

\bibitem{w12} 
\newblock V. Edemskiy and C. Wu,
\newblock {4-adic complexity of quaternary cyclotomic sequences and Ding-Helleseth sequences with period $pq$},
\newblock \emph{2022 IEEE International Symposium on Information Theory (ISIT)}, (2022), 372--377.

\bibitem{w13}
\newblock V. Edemskiy and Z. Chen,
\newblock {On the 4-adic complexity of the two-prime quaternary generator},
\newblock \emph{Journal of Applied Mathematics and Computing}, \textbf{68} (2022), 3565--3585.

\bibitem{w14}
\newblock T. Jiang and F. W. Fu,
\newblock {Some new classes of quaternary sequences with low autocorrelation property via two binary cyclotomic sequences},
\newblock \emph{Journal of Applied Mathematics and Computing}, \textbf{69} (2023), 689--706.

\bibitem{w15}
\newblock Z. Yang and P. H. Ke,
\newblock {Quaternary sequences with odd period and low autocorrelation},
\newblock \emph{Electronics letters}, \textbf{46} (2010), 1--2.

\bibitem{w16}
\newblock C. Zhang, X. Jing and Z. Xu,
\newblock {The linear complexity and 4-adic complexity of quaternary sequences with period $pq$},
\newblock \emph{Journal of Applied Mathematics and Computing}, \textbf{69} (2023), 2003--2017.

\bibitem{w17}
\newblock T. Cusick, C. Ding and A. Renvall,
\newblock {Stream ciphers and number theory},
\newblock \emph{Gulf Professional Publishing}, (2004).

\bibitem{w18}
\newblock T. Storer,
\newblock {Cyclotomy and difference sets},
\newblock \emph{Lectures in Advanced Mathematics}, (1967).

\bibitem{w19}
\newblock F. Sun, Q. Yue and X. Li,
\newblock {On the 2-adic complexity of cyclotomic binary sequences of order four},
\newblock \emph{Applicable Algebra in Engineering, Communication and Computing}, (2023), 1--19.

\end{thebibliography}
\end{document}